\documentclass[a4paper]{article}
\usepackage{pdfsync}
\usepackage[latin1]{inputenc}
\usepackage{graphicx}
\usepackage{amssymb,amsfonts,amsmath}
\usepackage{theorem}
\usepackage{hyperref}
\usepackage{graphicx}
\usepackage{multirow}
\usepackage{verbatim}

\usepackage{color}

\newtheorem{proposition}{Proposition}[section]
{\theorembodyfont{\rmfamily}

\newtheorem{remark}{Remark}[section]

}
\newenvironment{proof}[1][Proof]{\noindent\textbf{#1.} }{\newline \hspace*{\textwidth}\hspace*{-0,4cm} \rule{0.5em}{0.5em} \vspace{0,2cm}}

\begin{document}

\title{A note on the computation of geometrically defined relative velocities}

\author{Vicente J. Bol\'os\\
{\small Dpto. Matem\'aticas para la Econom\'{\i}a y la Empresa, Facultad de Econom\'{\i}a,}\\
{\small Universidad de Valencia. Avda. Tarongers s/n. 46022, Valencia,
Spain.}\\
{\small e-mail\textup{: \texttt{vicente.bolos@uv.es}}}}

\date{September 2011}

\maketitle

\begin{abstract}
We discuss some aspects about the computation of kinematic, spectroscopic, Fermi and astrometric relative velocities that are geometrically defined in general relativity. Mainly, we state that kinematic and spectroscopic relative velocities only depend on the 4-velocities of the observer and the test particle, unlike Fermi and astrometric relative velocities, that also depend on the acceleration of the observer and the corresponding relative position of the test particle, but only at the event of observation and not around it, as it would be deduced, in principle, from the definition of these velocities. Finally, we propose an open problem in general relativity that consists on finding intrinsic expressions for Fermi and astrometric relative velocities avoiding terms that involve the evolution of the relative position of the test particle. For this purpose, the proofs given in this paper can serve as inspiration.
\end{abstract}

%\vspace{10 pt}

%\noindent {\small \textbf{Keywords:} Relative velocity, relative position, observer, test particle}\\

\section{Introduction}

The concept of ``relative velocity'' of a distant test particle with respect to an observer is ambiguous in general relativity, in the sense that different coordinate systems and notions of simultaneity yield different results. This ambiguity led to consideration at the General Assembly of the International Astronomical Union (IAU), held in 2000 (see \cite{Soff03, Lind03}), introducing different definitions of ``radial velocity'' based on the Barycentric Celestial Reference System (BCRS).
However, a geometric concept of ``relative velocity'' should be intrinsic and independent from any coordinate system. Following this idea, four different geometric definitions were introduced in \cite{Bolos07}: \textit{kinematic}, \textit{Fermi}, \textit{spectroscopic} and \textit{astrometric} relative velocities. These four concepts each have full physical sense, and have proved to be useful in the study of properties of particular spacetimes \cite{KC10, Klein11, BK11} (see \cite{BK11} for a more detailed list of related works).
%Related work includes \cite{Carrera10b} for the study of observer-referred kinematics and dynamics, \cite{Teyss06} for the study of angular distances, \cite{Doran08} for a study of Schwarzschild black holes, and \cite{Lachieze06, Carrera10} for Doppler tracking.

In this paper we discuss some aspects about the computation of these relative velocities, and it is organized as follows. In Section \ref{sec2} we present the framework, establishing the notation and defining some necessary concepts, introducing in Section \ref{sec2.1} the four geometric concepts of relative velocity. In Section \ref{sec3} we develop the discussion, making special interest on those aspects concerning the geometric elements that are needed for the computation of the relative velocities. We further study this question for the Fermi and astrometric relative velocities, in special relativity (Section \ref{sec3.1}) and general relativity (Section \ref{sec3.2}). Finally, we give some concluding remarks in Section \ref{sec4}.

\section{Definitions and notation}
\label{sec2}

We work in a lorentzian spacetime manifold $\left(
\mathcal{M},g\right) $, with $c=1$ and $\nabla $ the Levi-Civita
connection, using the Landau-Lifshitz Spacelike Convention (LLSC).
We suppose that $\mathcal{M}$ is a convex normal neighborhood; thus, given two events $p$ and $q$ in
$\mathcal{M}$, there exists a unique geodesic joining them. The parallel transport from $q$ to $p$
along this geodesic is denoted by $\tau _{qp}$. If $\beta
:I\rightarrow \mathcal{M}$ is a curve with $I\subseteq \mathbb{R}$ a
real interval, we identify $\beta $ with the image $\beta I$
(that is a subset in $\mathcal{M}$), in order to simplify the notation.
Vector fields are denoted by uppercase letters and vectors (defined at a single point) are denoted by lowercase letters.
If $u$ is a vector, then $u^{\bot }$ denotes the
orthogonal space of $u$. The projection of a vector $v$ onto
$u^{\bot }$ is the projection parallel to $u$. Moreover, if $x$ is
a spacelike vector, then $\Vert x\Vert :=g\left( x,x\right) ^{1/2}$ is the modulus of
$x$. If $X$ is a vector field,
$X_p$ denotes the unique vector of $X$ in
$T_p\mathcal{M}$.

In general, we say that a timelike world line $\beta $ is an
\textit{observer} (or a \textit{test particle}). Nevertheless, we say that a
future-pointing timelike unit vector $u$ in $T_{p}\mathcal{M}$ is
an \textit{observer at $p$}, identifying it with its 4-velocity.

A \textit{light ray} is a lightlike geodesic $\lambda $. A \textit{light ray from }$q$\textit{\ to }$p$
is a light ray $\lambda $ such that $q,p\in \lambda $ and $p$ is in the causal future of $q$.

\subsection{Geometrically defined relative velocities}
\label{sec2.1}

Throughout the paper, we consider an observer $\beta $ and a test particle $\beta '$ (parameterized by their proper times) with 4-velocities $U$, $U'$ respectively.
Let $p$ be an event of $\beta $, and $u:=U_p$. We define $q_{\mathrm{s}}$ as the event of $\beta '$ such that there exists a spacelike geodesic $\psi $ orthogonal to $u$ joining $p$ and $q_{\mathrm{s}}$ (see Figure \ref{diagram}); analogously, let $q_{\ell}$ be the event of $\beta '$ such that there exists a light ray $\lambda $ from $q_{\ell}$ to $p$. We denote $u'_{\mathrm{s}}:=U'_{q_{\mathrm{s}}}$ and $u'_{\ell}:=U'_{q_{\ell}}$ in order to simplify the notation.

\begin{figure}[tbp]
\begin{center}
\includegraphics[width=0.33\textwidth]{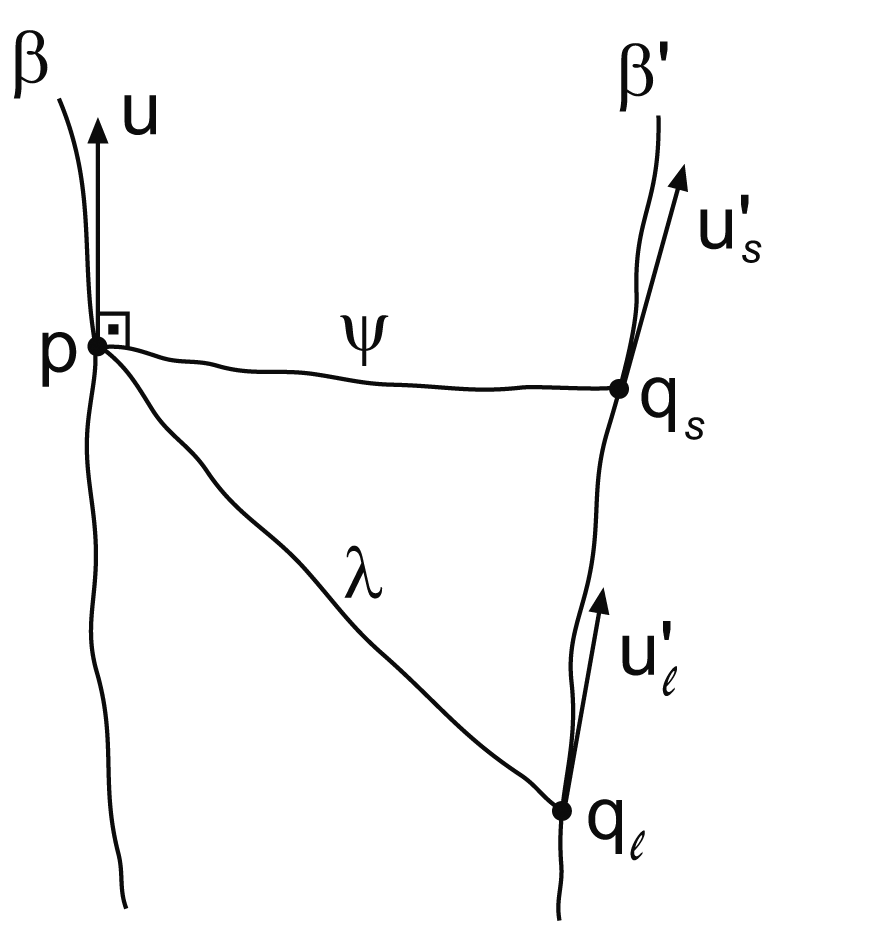}
\end{center}
\caption{Scheme of the elements involved in the study of the relative velocities of $\beta '$ with respect to $\beta $. The curve $\psi $ is a spacelike geodesic orthogonal to the 4-velocity of $\beta $ at $p$, denoted by $u$, and $\lambda $ is a light ray from $q_{\ell}$ to $p$. The vectors $u'_{\mathrm{s}}$ and $u'_{\ell}$ are the 4-velocities of $\beta '$ at $q_{\mathrm{s}}$ and $q_{\ell}$ respectively.} \label{diagram}
\end{figure}

The \textit{kinematic relative velocity of $u'_{\mathrm{s}}$ with respect to $u$} is the vector
\begin{equation}
\label{vkin}
v_{\mathrm{kin}}:=\frac{1}{-g\left( \tau _{q_{\mathrm{s}} p}u'_{\mathrm{s}},u\right) }\tau _{q_{\mathrm{s}} p}u'_{\mathrm{s}}-u.
\end{equation}
Analogously, the \textit{spectroscopic relative velocity of $u'_{\ell}$ with respect to (or observed by) $u$} is the vector
\begin{equation}
\label{vspec}
v_{\mathrm{spec}}:=\frac{1}{-g\left( \tau _{q_{\ell} p}u'_{\ell},u\right) }\tau _{q_{\ell} p}u'_{\ell}-u.
\end{equation}
Note that $v_{\mathrm{kin}}$ and $v_{\mathrm{spec}}$ are both spacelike and orthogonal to $u$.
Varying $p$ along $\beta $, we construct the vector fields $V_{\mathrm{kin}}$ and $V_{\mathrm{spec}}$ defined on $\beta $, representing the \textit{kinematic} and the \textit{spectroscopic relative velocity of $\beta '$ with respect to $\beta $}, respectively (see \cite[Definitions 3, 4, 10 and 11]{Bolos07}).

The \textit{relative position of $q_{\mathrm{s}}$ with respect to $u$} is the vector
\begin{equation}
\label{eqs}
s:=\exp_p^{-1} q_{\mathrm{s}}=\log _p q_{\mathrm{s}},
\end{equation}
where $\exp _p$ is the exponential map\footnote{Given $v\in T_p\mathcal{M}$, $\exp _p v:=\gamma _v(1)$ where $\gamma _v$ is a geodesic starting at $p$ with initial tangent vector $v$.} on $T_p\mathcal{M}$ and $\log _p$ is its inverse (note that the log map is well-defined because we work in a convex normal neighborhood). On the other hand, the \textit{observed relative position of $q_{\ell}$ with respect to (or observed by) $u$} is the projection of $\log _p q_{\ell}$ onto $u^{\bot }$, i.e. it is the vector
\begin{equation}
\label{eqsobs}
s_{\mathrm{obs}}:=\log _p q_{\ell}+g\left( \log _p q_{\ell},u\right) u.
\end{equation}
Note that $s$ and $s_{\mathrm{obs}}$ are both spacelike and orthogonal to $u$. Varying $p$ along $\beta $, we construct the vector fields $S$ and $S_{\mathrm{obs}}$ defined on $\beta $, representing the \textit{relative position} and the \textit{observed relative position of $\beta '$ with respect to $\beta $}, respectively (see \cite[Definitions 1, 2, 8 and 9]{Bolos07}).

The \textit{Fermi relative velocity of $\beta '$ with respect to $\beta $} is the projection of $\nabla
_{U}S$ onto $U^{\bot }$, i.e. it is the vector field
\begin{equation}
\label{vfermi}
V_{\mathrm{Fermi}}:=\nabla _{U}S+g\left( \nabla _{U}S,U\right) U=\nabla _{U}S-g\left( S,\nabla _{U}U\right) U,
\end{equation}
defined on $\beta $ (see \cite[Definition 5 and Proposition 1]{Bolos07}).
Analogously, the \textit{astrometric relative velocity of $\beta '$ with respect to $\beta $} is the projection of $\nabla
_{U}S_{\mathrm{obs}}$ onto $U^{\bot }$, i.e. it is the vector field
\begin{equation}
\label{vast}
V_{\mathrm{ast}}:=\nabla _{U}S_{\mathrm{obs}}+g\left( \nabla _{U}S_{\mathrm{obs}},U\right) U=\nabla _{U}S_{\mathrm{obs}}-g\left( S_{\mathrm{obs}},\nabla _{U}U\right) U,
\end{equation}
defined on $\beta $ (see \cite[Definition 12 and Proposition 5]{Bolos07}). Note that both relative velocities are spacelike and orthogonal to $U$. In order to complete the notation that we are going to use, we define the vectors $v_{\mathrm{Fermi}}:=V_{\mathrm{Fermi}\, p}$ and $v_{\mathrm{ast}}:=V_{\mathrm{ast}\, p}$; moreover, throughout the paper we are going to denote $s:=S_p$, $s_{\mathrm{obs}}:=S_{\mathrm{obs}\, p}$, $v_{\mathrm{kin}}:=V_{\mathrm{kin}\, p}$ and $v_{\mathrm{spec}}:=V_{\mathrm{spec}\, p}$.

\section{Computation of relative velocities}
\label{sec3}

The kinematic relative velocity at $p$ depends only on the 4-velocities of the observer at $p$ and the test particle at $q_{\mathrm{s}}$, i.e. $u$ and $u'_{\mathrm{s}}$ respectively (see Figure \ref{diagram}), and according to \eqref{vkin}, for its computation we have to find the vector $\tau_{q_{\mathrm{s}} p}u'_{\mathrm{s}}$. In practice, we can find a vector field $X$ which is tangent and parallel to the geodesic $\psi$ and then, taking into account that the connection is metric-compatible (i.e. the metric tensor is parallelly transported along geodesics), equations $g\left( X_p,\tau_{q_{\mathrm{s}} p}u'_{\mathrm{s}}\right) =g\left( X_{q_{\mathrm{s}}},u'_{\mathrm{s}}\right) $ and $g\left( \tau_{q_{\mathrm{s}} p}u'_{\mathrm{s}},\tau_{q_{\mathrm{s}} p}u'_{\mathrm{s}}\right) =g\left( u'_{\mathrm{s}},u'_{\mathrm{s}}\right) =-1$ are very useful (see examples in \cite{Bolos07,KC10,BK11}). This is analogous for the spectroscopic relative velocity, taking into account \eqref{vspec} and considering $q_{\ell}$, $u'_{\ell}$, $\lambda$ instead of $q_{\mathrm{s}}$, $u'_{\mathrm{s}}$, $\psi$.

On the other hand, according to \eqref{vfermi}, the Fermi relative velocity at $p$ apparently also depends on $\left( \nabla _U S\right) _p$ and $\left( \nabla _U U\right) _p$ (interpreted as the \textit{acceleration} of the observer at $p$). Consequently, it would not suffice to know the vectors $u$ and $s$ (at $p$), but it would be necessary to know the behavior of the corresponding vector fields, $U$ and $S$, around $p$ (concretely, at the intersection of a neighborhood of $p$ and the observer $\beta $). This is analogous for the astrometric relative velocity, taking into account \eqref{vast} and considering $S_{\mathrm{obs}}$ instead of $S$.

In the following sections, we give some results that weaken this condition and let us to compute the Fermi and astrometric relative velocities knowing the acceleration of the observer and $S$ or $S_{\mathrm{obs}}$ \emph{only at $p$} (i.e. $s$ or $s_{\mathrm{obs}}$).

\subsection{Special relativity}
\label{sec3.1}

In this section, we work in the Minkowski spacetime, considering that all the tangent spaces are canonically identified by means of parallel transport. Moreover, the Minkowski spacetime has an affine structure and, given two events $p,q\in \mathcal{M}$, the vector which joins $p$ and $q$ (i.e. $\log _p q$) is given by $q-p$.

The goal is to find expressions for $V_{\mathrm{Fermi}}$ and $V_{\mathrm{ast}}$ in terms of $U$, $\nabla _U U$, $U'$, $S$ and $S_{\mathrm{obs}}$, avoiding $\nabla _U S$, $\nabla _U S_{\mathrm{obs}}$, or any term involving the evolution of $S$ and $S_{\mathrm{obs}}$ around $p$.
These expressions were previously found in the proofs of \cite[Propositions 8 and 9]{Bolos07} (which were given in terms of $V_{\mathrm{kin}}$ and $V_{\mathrm{spec}}$), but we present here new proofs in a more geometric way that are susceptible to be extended to general relativity and serve us to prove the main result of this work: Proposition \ref{propgen1}.

\begin{proposition}
\label{propmink1} It holds
\begin{equation}
\label{fmink1_3}
V_{\mathrm{Fermi}} =\left( 1+g\left( S,\nabla _{U}U\right)
\right) \left( \frac{1}{-g\left( U',U\right)
}U'-U\right) ,
\end{equation}
where $V_{\mathrm{Fermi}}$, $U$, $S$, $\nabla _U U$ are evaluated at an event $p$ of $\beta $, and $U'$ is evaluated at the corresponding event $q_{\mathrm{s}}$ of $\beta '$ (see Figure \ref{diagram}).
\end{proposition}

\begin{proof}
Let $p=\beta \left( \tau \right) $ be an event of
$\beta $ (where $\tau $ is the proper time of $\beta $ at $p$), and let $u\left( \tau \right) $ be the 4-velocity of $\beta
$ at $p$. Considering the corresponding event $q_{\mathrm{s}}$ of $\beta '$, the relative position of $q_{\mathrm{s}}$ with
respect to $u(\tau )$, denoted by $s(\tau )$, is given by $q_{\mathrm{s}}-p$. So, if $\tau '\left( \tau \right) $ is the proper time of $\beta '$ at $q_{\mathrm{s}}$ and $u'_{\mathrm{s}}=u'\left( \tau '(\tau )\right) $ is the 4-velocity of $\beta '$ at $q_{\mathrm{s}}$, by \eqref{eqs} we have
\begin{equation}
\label{eqm1}
s\left( \tau \right) =q_{\mathrm{s}}-p=\beta' \left( \tau '\left( \tau \right) \right)-\beta \left( \tau \right) \,\Longrightarrow \, \dot{s}=u'_{\mathrm{s}} \,\dot{\tau }'-u,
\end{equation}
where the overdot represents differentiation with respect to $\tau $.
On the other hand
\begin{equation}
\label{eqm2}
g\left( s,u\right) =0 \,\Longrightarrow \, g\left( \dot{s},u\right) +g\left( s,\dot{u}\right) =0.
\end{equation}
Applying \eqref{eqm1} in \eqref{eqm2} we have
\begin{equation}
\label{eqm3}
g\left( u'_{\mathrm{s}} \,\dot{\tau }'-u,u\right) +g\left( s,\dot{u}\right) =0 \,\Longrightarrow \, \dot{\tau }'=\frac{1+g\left( s,\dot{u}\right) }{-g\left( u'_{\mathrm{s}} ,u\right) },
\end{equation}
and then, combining \eqref{eqm1} with \eqref{eqm3}, we obtain
\begin{equation}
\label{fmink1}
\dot{s}=\frac{1+g\left( s,\dot{u}\right) }{-g\left( u'_{\mathrm{s}} ,u\right) }u'_{\mathrm{s}} -u.
\end{equation}

Using vector fields, from \eqref{fmink1} we have
\begin{equation}
\label{fmink1_2}
\nabla _{U}S=\frac{1+g\left( S,\nabla _{U}U\right) }{-g\left( U',U\right) }U'-U,
\end{equation}
where $U$, $S$, $\nabla _U U$, $\nabla _U S$ are evaluated at $p$, and $U'$ is evaluated at the corresponding event $q_{\mathrm{s}}$.
So, applying \eqref{fmink1_2} in \eqref{vfermi},
the Fermi relative velocity $V_{\mathrm{Fermi}}$ of $\beta
'$ with respect to $\beta $ is given by \eqref{fmink1_3}.
\end{proof}

Next, we give an analogous result for the astrometric relative velocity, in the case that the observer and the test particle do not intersect. Otherwise, we have to take into account Remark \ref{rem31}.

\begin{proposition}
\label{propmink2}
If $S_{\mathrm{obs}}$ does not vanish, then
\begin{equation}
\label{fmink2_3}
V_{\mathrm{ast}} = \frac{-g\left( U',U\right) }{g\left( U',\frac{S_{\mathrm{obs}}}{\| S_{\mathrm{obs}}\| }-U\right) }\left( \frac{1}{-g\left( U',U\right)
}U'-U\right) +\| S_{\mathrm{obs}}\| \nabla _U U,
\end{equation}
where $V_{\mathrm{ast}}$, $U$, $S_{\mathrm{obs}}$, $\nabla _U U$ are evaluated at an event $p$ of $\beta $, and $U'$ is evaluated at the corresponding event $q_{\ell}$ of $\beta '$ (see Figure \ref{diagram}).
\end{proposition}

\begin{proof}
Let $p=\beta \left( \tau \right) $ be an event of
$\beta $ (where $\tau $ is the proper time of $\beta $ at $p$), and let $u\left( \tau \right) $ be the 4-velocity of $\beta
$ at $p$. Considering the corresponding event $q_{\mathrm{\ell}}$ of $\beta '$, the relative position of $q_{\ell}$ observed by $u(\tau )$, denoted by $s_{\textrm{obs}}(\tau )$, is the projection of $q_{\ell}-p$ onto $u^{\bot }(\tau )$. If $\tau '\left( \tau \right) $ is the proper time of $\beta '$ at $q_{\ell}$, by \eqref{eqsobs} we have
\begin{equation}
\label{eqm21}
s_{\textrm{obs}}\left( \tau \right) =q_{\ell}-p+g\left( q_{\ell}-p,u(\tau )\right) u(\tau )=\beta' \left( \tau '\left( \tau \right) \right)-\beta \left( \tau \right) +\| s_{\textrm{obs}}\left( \tau\right) \| u(\tau ),
\end{equation}
where $\| s_{\textrm{obs}}(\tau )\|$ is the affine distance from $p$ to $q_{\ell}$ observed by $u(\tau )$ (see \cite[Definition 13]{Bolos07}).
If $u'_{\ell}=u'\left( \tau '(\tau )\right) $ is the 4-velocity of $\beta '$ at $q_{\ell}$, from \eqref{eqm21} we obtain
\begin{equation}
\label{eqm22}
\dot{s}_{\textrm{obs}}=u'_{\ell}\, \dot{\tau }'-u+g\left( \dot{s}_{\textrm{obs}},\frac{s_{\textrm{obs}}}{\| s_{\textrm{obs}}\| }\right) u+\| s_{\textrm{obs}}\| \dot{u},
\end{equation}
where the overdot represents differentiation with respect to $\tau $. Taking into account that $g\left( s_{\textrm{obs}},u\right) =0$ and \eqref{eqm22}, we have
\begin{equation}
\label{eqm23}
g\left( \dot{s}_{\textrm{obs}},\frac{s_{\textrm{obs}}}{\| s_{\textrm{obs}}\| }\right) =g\left( u'_{\ell}\, \dot{\tau }'+\| s_{\textrm{obs}}\| \dot{u},\frac{s_{\textrm{obs}}}{\| s_{\textrm{obs}}\| }\right) =\dot{\tau }'g\left( u'_{\ell} ,\frac{s_{\textrm{obs}}}{\| s_{\textrm{obs}}\| }\right) +g\left( \dot{u},s_{\textrm{obs}}\right) ,
\end{equation}
and hence, by \eqref{eqm22} and \eqref{eqm23} we obtain
\begin{equation}
\label{eqm24}
\dot{s}_{\textrm{obs}}=u'_{\ell}\, \dot{\tau }'+\left( \dot{\tau }'g\left( u'_{\ell} ,\frac{s_{\textrm{obs}}}{\| s_{\textrm{obs}}\| }\right) +g\left( \dot{u},s_{\textrm{obs}}\right) -1\right) u+\| s_{\textrm{obs}}\| \dot{u} .
\end{equation}
On the other hand
\begin{equation}
\label{eqm25}
g\left( s_{\textrm{obs}},u\right) =0 \,\Longrightarrow \, g\left( \dot{s}_{\textrm{obs}},u\right) +g\left( s_{\textrm{obs}},\dot{u}\right) =0.
\end{equation}
Applying \eqref{eqm24} in \eqref{eqm25}, and taking into account that $g\left( \dot{u},u\right) =0$, we find
\begin{equation}
\label{eqm26}
\dot{\tau }'=\frac{1}{g\left( u'_{\ell} ,\frac{s_{\textrm{obs}}}{\| s_{\textrm{obs}}\| }-u\right) },
\end{equation}
and then, combining \eqref{eqm24} with \eqref{eqm26}, we obtain
\begin{equation}
\label{fmink2}
\dot{s}_{\textrm{obs}} = \frac{1}{g\left( u'_{\ell} ,\frac{s_{\textrm{obs}}}{\| s_{\textrm{obs}}\| }-u\right) }\left( u'_{\ell} +g\left( u'_{\ell} ,u\right) u\right) +g\left( s_{\textrm{obs}},\dot{u}\right) u+\| s_{\textrm{obs}}\| \dot{u} .
\end{equation}

Using vector fields, from \eqref{fmink2} we have
\begin{equation}
\label{fmink2_2}
\nabla _U S_{\textrm{obs}}=\frac{1}{g\left( U',\frac{S_{\textrm{obs}}}{\| S_{\textrm{obs}}\| }-U\right) }\left( U'+g\left( U',U\right) U\right) +g\left( S_{\textrm{obs}},\nabla _U U\right) U+\| S_{\textrm{obs}}\| \nabla _U U,
\end{equation}
where $U$, $S_{\textrm{obs}}$, $\nabla _U U$, $\nabla _U S_{\textrm{obs}}$ are evaluated at $p$, and $U'$ is evaluated at the corresponding event $q_{\ell}$.
So, applying \eqref{fmink2_2} in \eqref{vast},
the astrometric relative velocity $V_{\mathrm{ast}}$ of $\beta
'$ with respect to $\beta $ is given by \eqref{fmink2_3}.
\end{proof}

\begin{remark}
\label{rem31}
If $S_{\mathrm{obs}}$ vanishes at $p$ (i.e. $s_{\mathrm{obs}}=0$) then $\beta $ and $\beta '$ intersect at $p$. In this case, if $u$ and $u'$ are the 4-velocities of $\beta $ and $\beta '$ at $p$ respectively, it is easy to prove that $v_{\mathrm{ast}}=\frac{1}{1\pm \| v\| }v$, where $v:=\frac{1}{-g\left( u',u\right) }u'-u$ is the usual relative velocity of $u'$ observed by $u$ (that coincides with $v_{\mathrm{kin}}$, $v_{\mathrm{Fermi}}$ and $v_{\mathrm{spec}}$, see \cite{Bolos07}). The ``$+$'' or ``$-$'' sign is taken for a ``leaving'' or ``arriving'' test particle respectively.
\end{remark}

%If the observer $\beta $ is geodesic then $\nabla _U U=0$ and according to \eqref{fmink1_3}, $V_{\mathrm{Fermi}}$ does not depend on $S$. Moreover, according to \eqref{fmink2_3}, $V_{\mathrm{ast}}$ only depends on $S_{\mathrm{obs}}$ through its direction (i.e. the normalized vector field $S_{\mathrm{obs}}/\| S_{\textrm{obs}}\| $), and hence, it does not depend on $\| S_{\textrm{obs}}\| $ (the affine distance from the test particle $\beta '$ to $\beta $ observed by $\beta $).

\subsection{General relativity}
\label{sec3.2}

As consequence of Propositions \ref{propmink1} and \ref{propmink2} we state that, in special relativity, $v_{\mathrm{Fermi}}$ and $v_{\mathrm{ast}}$ can be computed in terms of $p$, $q_{\mathrm{s}}$, $q_{\ell }$, $u$, $\left( \nabla _U U\right) _p$, $u'_{\mathrm{s}}$, $u'_{\ell}$, $s$ and $s_{\mathrm{obs}}$. Hence, we do not need to know $S$ or $S_{\mathrm{obs}}$ around $p$ for computing these relative velocities. Next, we are going to generalize this result.

\begin{proposition}
\label{propgen1} In general relativity,
\begin{itemize}
\item $v_{\mathrm{Fermi}}$ is completely determined by $p$, $q_{\mathrm{s}}$, $u$, $\left( \nabla _U U\right) _p$, $u'_{\mathrm{s}}$ and $s$.
\item $v_{\mathrm{ast}}$ is completely determined by $p$, $q_{\ell }$, $u$, $\left( \nabla _U U\right) _p$, $u'_{\ell}$ and $s_{\mathrm{obs}}$.
\end{itemize}
\end{proposition}

\begin{proof}
First, we are going to deal with the Fermi relative velocity, generalizing the steps of the proof of Proposition \ref{propmink1}. Let $p=\beta \left( \tau \right) $ be an event of
$\beta $ (where $\tau $ is the proper time of $\beta $ at $p$), and let $u\left( \tau \right) $ be the 4-velocity of $\beta
$ at $p$. From \eqref{eqs}, the relative position of the corresponding $q_{\mathrm{s}}$ with respect to $u(\tau )$ is given by
\begin{equation}
\label{eqgen1}
s(\tau )=\log \left( p,q_{\mathrm{s}}\right) =\exp ^{-1}_p q_{\mathrm{s}}=\exp ^{-1}_{\beta (\tau)} \beta '\left( \tau '(\tau)\right) ,
\end{equation}
where $\tau '\left( \tau \right) $ is the proper time of $\beta '$ at $q_{\mathrm{s}}$ and we use the notation $\log \left( p,q_{\mathrm{s}}\right) :=\log _p q_{\mathrm{s}}$ for convenience.
Since we work in a convex normal neighborhood, we can use a local coordinate system $\left( x^0,x^1,x^2,x^3\right) $ containing $p$ and $q_{\mathrm{s}}$, and hence
\begin{equation}
\label{eqgen2}
\left( \nabla _U S\right) _p=\dot{s}+u^j s^k \Gamma ^i_{jk}(p)\left. \frac{\partial }{\partial x^i}\right| _p,
\end{equation}
where the overdot represents differentiation with respect to $\tau $.
Taking into account \eqref{eqgen1} and the fact that $u'_{\mathrm{s}}=u'\left( \tau '(\tau )\right) $ is the 4-velocity of $\beta '$ at $q_{\mathrm{s}}$, by the chain rule we have
\begin{equation}
\label{eqgen3}
\dot{s}=\left( \left. \frac{\partial f_1^i}{\partial x^j}\right| _p u^j+\left. \frac{\partial f_2^i}{\partial x^j}\right| _{q_{\mathrm{s}}} {u'_{\mathrm{s}}}^j \dot{\tau}'\right) \left. \frac{\partial }{\partial x^i}\right| _p,
\end{equation}
where $f_1:=\log \left( \underline{\,\,\,\,} ,q_{\mathrm{s}}\right) $ and $f_2:=\log \left( p,\underline{\,\,\,\,} \right) =\log _p=\exp _p^{-1}$. Note that the derivatives of these functions are completely determined by the coordinates of $p$ and $q_{\mathrm{s}}$ and so, for our purposes, we do not need to compute them. Nevertheless, they can be computed by means of the Jacobi fields theory (see, for example, \cite{Cheeger75,doCarmo92}).

On the other hand
\begin{equation}
\label{eqgen4}
g\left( S,U\right) =0 \,\Longrightarrow \, g\left( \nabla _U S,U\right) +g\left( S,\nabla _U U\right) =0.
\end{equation}
Then, applying \eqref{eqgen2} and \eqref{eqgen3} in \eqref{eqgen4}, we can solve $\dot{\tau}'$ in terms of the coordinates of $p$, $q_{\mathrm{s}}$, $u$,  $\left( \nabla _U U\right) _p$, $u'_{\mathrm{s}}$ and $s$. Hence, taking this into account jointly with \eqref{eqgen2}, \eqref{eqgen3} and the expression of the Fermi relative velocity given in \eqref{vfermi}, the result holds.

With respect to the astrometric relative velocity, it can be proved analogously, taking into account \eqref{eqsobs} and generalizing the steps of the proof of Proposition \ref{propmink2}.
\end{proof}

There is an open problem that consists on finding intrinsic expressions (in a coordinate-free language) for the Fermi and astrometric relative velocities, analogous to those given in Propositions \ref{propmink1} and \ref{propmink2}, i.e. in terms of $U$, $\nabla _U U$, $U'$, $S$ and $S_{\mathrm{obs}}$, avoiding $\nabla _U S$, $\nabla _U S_{\mathrm{obs}}$, or any term involving the evolution of $S$ and $S_{\mathrm{obs}}$ around $p$. It is a hard geometric problem, but it would be very useful for the interpretation and computation of these relative velocities.

\section{Concluding remarks}
\label{sec4}

\begin{figure}[tbp]
\begin{center}
\includegraphics[width=0.7\textwidth]{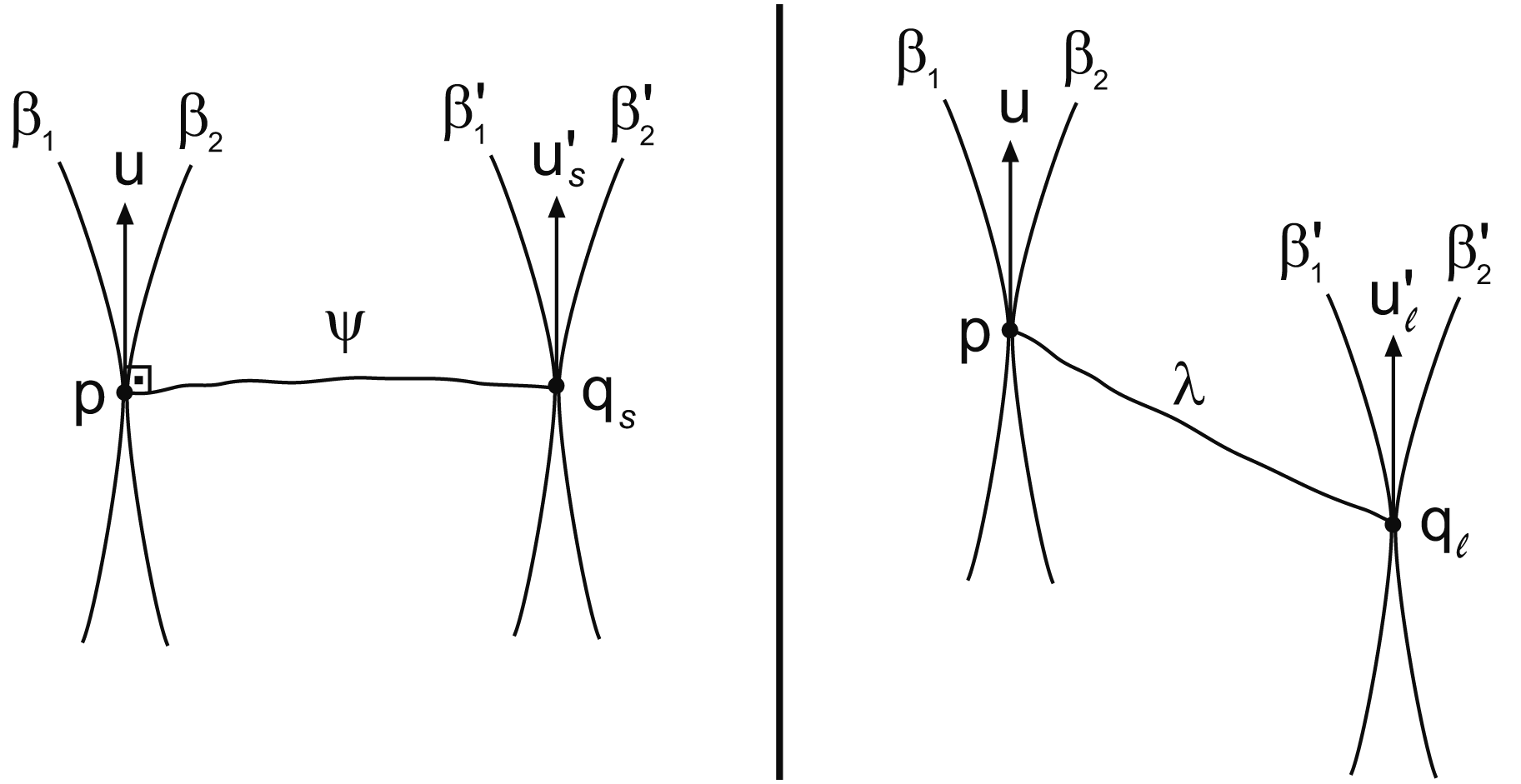}
\end{center}
\caption{The observers $\beta _1$ and $\beta _2$ have the same 4-velocity $u$ at $p$. In the same way, the test particles $\beta '_1$ and $\beta '_2$ have the same 4-velocity at $q_{\mathrm{s}}$ (left) or $q_{\ell}$ (right). In these diagrams, $\psi$ is a spacelike geodesic orthogonal to $u$ at $p$, and $\lambda $ is a light ray from $q_{\ell}$ to $p$.}
\label{diagram2}
\end{figure}

Taking into account \eqref{vkin} and \eqref{vspec}, if we study the kinematic and spectroscopic relative velocities of the test particles with respect to the observers represented in the Figure \ref{diagram2}, we have
\begin{equation*}
\begin{array}{lclclcl}
v_{\mathrm{kin}\,1,1} & = & v_{\mathrm{kin}\,1,2} & = & v_{\mathrm{kin}\,2,1} & = & v_{\mathrm{kin}\,2,2} \, ,\\
v_{\mathrm{spec}\,1,1} & = & v_{\mathrm{spec}\,1,2} & = & v_{\mathrm{spec}\,2,1} & = & v_{\mathrm{spec}\,2,2} \, ,
\end{array}
\end{equation*}
where the first subindex refers to the observer ($\beta _1$ or $\beta _2$) and the second to the test particle ($\beta '_1$ or $\beta ' _2$).

Moreover, in principle (i.e. taking into account only the definitions \eqref{vfermi} and \eqref{vast}), for the Fermi and astrometric relative velocities in Figure \ref{diagram2} we have
\begin{equation*}
\begin{array}{lclclcl}
v_{\mathrm{Fermi}\,1,1} & \neq & v_{\mathrm{Fermi}\,1,2} & \neq & v_{\mathrm{Fermi}\,2,1} & \neq & v_{\mathrm{Fermi}\,2,2}\, , \\
v_{\mathrm{ast}\,1,1} & \neq & v_{\mathrm{ast}\,1,2} & \neq & v_{\mathrm{ast}\,2,1} & \neq & v_{\mathrm{ast}\,2,2}\, ,
\end{array}
\end{equation*}
as it is discussed at the beginning of Section \ref{sec3}. With regard to the acceleration of the observer at $p$, $\left( \nabla _U U\right) _p$, it is shown in Propositions \ref{propmink1} and \ref{propmink2} that, indeed, we have to know it (even in special relativity). The fact that the relative velocity of a test particle depends on the acceleration of the observer is not intuitive, but it is acceptable in the framework of relativity. On the other hand, the fact that $\beta '_1$ and $\beta '_2$ have different relative velocities is even less intuitive and less acceptable. However, as it is proved in Proposition \ref{propgen1}, this last statement is not true and we just need to know $s$ or $s_{\mathrm{obs}}$ (only at $p$). Therefore, we actually have
\begin{equation*}
\begin{array}{lclclcl}
v_{\mathrm{Fermi}\,1,1} & = & v_{\mathrm{Fermi}\,1,2} & \neq & v_{\mathrm{Fermi}\,2,1} & = & v_{\mathrm{Fermi}\,2,2}\, , \\
v_{\mathrm{ast}\,1,1} & = & v_{\mathrm{ast}\,1,2} & \neq & v_{\mathrm{ast}\,2,1} & = & v_{\mathrm{ast}\,2,2}\, .
\end{array}
\end{equation*}

\section*{Acknowledgments}

I would like to thank Prof. Vicente Miquel and Prof. Juan Antonio Navarro for their valuable help and comments.

%\section*{References}


\begin{thebibliography}{99}

\bibitem{Soff03} M. Soffel, \textit{et al}. The IAU 2000 resolutions for astrometry, celestial mechanics and metrology in the relativistic framework: explanatory supplement. \textit{Astron. J.} \textbf{126} (2003), 2687--2706 (arXiv:\href{http://arxiv.org/abs/astro-ph/0303376}{astro-ph/0303376}).

\bibitem{Lind03} L. Lindegren, D. Dravins. The fundamental definition of `radial velocity'. \textit{Astron. Astrophys.} \textbf{401} (2003), 1185--1202 (arXiv:\href{http://arxiv.org/abs/astro-ph/0302522}{astro-ph/0302522}).

\bibitem{Bolos07} V. J. Bol\'os. Intrinsic definitions of ``relative velocity'' in general relativity. \textit{Commun. Math. Phys.} \textbf{273} (2007), 217--236 (arXiv:\href{http://arxiv.org/abs/gr-qc/0506032}{gr-qc/0506032}).

\bibitem{KC10} D. Klein, P. Collas. Recessional velocities and Hubble's law in Schwarzschild-de Sitter space. \textit{Phys. Rev.} D \textbf{81} (2010), 063518 (arXiv:\href{http://arxiv.org/abs/1001.1875}{1001.1875}).

\bibitem{Klein11} D. Klein, E. Randles. Fermi coordinates, simultaneity, and expanding space in Robertson-Walker cosmologies. \textit{Ann. Henri Poincar\'e} \textbf{12} (2011), 303--328 (arXiv:\href{http://arxiv.org/abs/1010.0588}{1010.0588}).

\bibitem{BK11} V. J. Bol\'os, D. Klein. Relative velocities for radial motion in expanding Robertson-Walker spacetimes. Preprint (2011), (arXiv:\href{http://arxiv.org/abs/1106.3859}{1106.3859}).

\bibitem{Cheeger75} J. Cheeger, D. G. Ebin. Comparison Theorems in Riemannian Geometry. North-Holland Publishing Company, Amsterdam (1975).

\bibitem{doCarmo92} M. P. do Carmo. Riemannian Geometry. Birkh\"auser, Boston (1992).

\end{thebibliography}
\end{document}